\documentclass[pra,superscriptaddress,,nofootinbib]{revtex4}
\usepackage{graphicx}
\usepackage{amsmath}
\usepackage{amssymb}
\usepackage{float}
\usepackage{algorithm}

\floatname{algorithm}{Protocol}
\usepackage{bm}
\usepackage{bbm}
\usepackage{amsthm}

\usepackage{color}

\def\EQ#1{\begin{eqnarray}#1\end{eqnarray}}
\def\Tr{{\mathrm{Tr}}}
\newcommand{\djj}{d\kern-0.4em\char"16\kern-0.1em}
\newcommand{\beq}{\begin{equation}}
\newcommand{\eeq}{\end{equation}}
\newcommand{\bea}{\begin{eqnarray}}
\newcommand{\eea} {\end{eqnarray}}

\newtheorem{lemma}{Lemma}

\newtheorem{prop}{Proposition}\def\PRO{\begin{prop}}\def\ORP{\end{prop}}
\newtheorem{coro}{Corollary}\def\COR{\begin{coro}}\def\ROC{\end{coro}}
\newtheorem{theo}{Theorem}\def\TH{\begin{theo}}\def\HT{\end{theo}}
\def\TH{\begin{theo}}\def\HT{\end{theo}}
\newtheorem{defi}[prop]{Definition}\def\DE{\begin{defi}}\def\ED{\end{defi}}
\newtheorem{lemme}[prop]{Lemma}\def\LE{\begin{lemme}}\def\EL{\end{lemme}}

\def\co{\mathcal{C}}
\def\a{\alpha}
\begin{document}
\title{Contrary Inferences in Consistent Histories and a Set Selection Criterion}

\author{Petros Wallden}
\email{petros.wallden@hw.ac.uk}
\affiliation{SUPA, School of Engineering and Physical Sciences, Heriot-Watt University, Edinburgh EH14 1AS, UK}
\affiliation{Physics Department, University of Athens, Panepistimiopolis 157-71, Ilisia Athens, Greece}

\begin{abstract}
The best developed formulation of closed system quantum theory that handles multiple-time statements, is the consistent (or decoherent) histories approach. The most important weaknesses of the approach is that it gives rise to many different consistent sets, and it has been argued that a complete interpretation should be accompanied with a natural mechanism leading to a (possibly) unique preferred consistent set. The existence of multiple consistent sets becomes more problematic because it allows the existence of contrary inferences \cite{Kent 1997}.
 We analyse the conceptual difficulties that arise from  the existence of multiple consistent sets and provide a
suggestion for a natural set selection criterion. This criterion does not lead
to a unique physical consistent set, however it evades the existence of consistent
sets with contrary inferences.
 The criterion is based on the concept of preclusion and the requirement that probability one propositions and their inferences should be non-contextual. The allowed consistent sets turn-out to be compatible with coevents which are the ontology of an alternative, histories based, formulation \cite{coevent1,coevent2,coevents review}.
\end{abstract}

\maketitle
\section{Introduction}
The construction of a theory of quantum gravity can be considered as the greatest challenge of modern physics. Among a great number of researchers it is believed that  a better understanding or a new formulation of quantum theory might be a necessary step in the search for quantum gravity. In the same time, concepts and ideas from quantum gravity can provide guiding principles for research in quantum foundations.

One guiding principle that quantum gravity may suggest, is the use of spacetime formulations of quantum theory (Dirac \cite{Dirac}, Feynman \cite{Feynman}). The different nature that space and time have in quantum theory is in conflict with the spirit of general relativity and leads to the problem of time in canonical approaches to quantum gravity. Therefore, \emph{histories} or \emph{path integral} formulations could be better suited for quantum gravity. Moreover, histories formulations can be casted in an observer independent way. The universe is the archetypical example of a closed system and an observer independent formulation is necessary for Quantum Cosmology.

However, obtaining a complete interpretation of quantum theory using histories as starting point is not easy. The best developed such approach is the consistent histories \cite{CH1,CH2,CH3,CH4} approach. Consistent histories have received several critical opinions \cite{ACH1,Kent 1997,ACH2,ACH3} and it is understood that the formalism needs further additions to appear as a complete viable alternative. Other attempts for a histories formulation include linearly positive histories \cite{GoPa,Hartle 2004}, the quantum measure theory \cite{quantum measure 1,quantum measure 2} and the development that followed those papers leading to the coevents formulation \cite{coevent1,coevent2,coevents review}. In this paper we will reexamine the most striking conceptual problem that stems from the consistent histories formalism, namely the existence of multiple (possibly incompatible) consistent sets, and suggest a possible solution.

We will first examine the concept of contextuality and probabilities from the perspective of a realist interpretation of standard quantum theory  in section \ref{Section contextuality}. Then we will briefly review the consistent histories approach in section \ref{Section consistent histories}. In section \ref{Section contrary inferences} we will state how contrary inferences appear in consistent histories and give an alternative derivation using the concept of zero covers which stresses some of the conceptual difficulties that arise. We will give two examples of zero covers in section \ref{Section examples}. In section \ref{Section PCS} we will give the proposed set selection criterion using the concept of preclusive consistent set. In section \ref{Section coevents PCS} we will introduce the concept of coevents and confirm that the criterion suggested is compatible with this ontology. The examples will be reexamined in section \ref{Section examples 2} and we conclude in section \ref{Section discussion}.

\section{Contexutality and probabilities}\label{Section contextuality}

Standard quantum theory has been proven to be contextual by the theorems of Gleason \cite{Gleason} and Kochen and Specker \cite{KS}.  However, this does not rule out ontological (hidden variables) models, provided they have this property. The de Broglie-Bohm theory \cite{de Broglie-Bohm} provides an example of Contextual Hidden Variables Theory (CHVT). CHVT are generally classified in two types depending on how the different contexts are determined \cite{Shimony 1984}. First are the \emph{environmental} CHVT, where the context is determined for example, from the specific apparatus or environment that is realised. For example de Broglie-Bohm theory is considered by some as environmental CHVT. Second are the \emph{algebraic} CHVT, where the possible contexts are maximal Boolean sub-algebras of projection operators \cite{Algebraic CHVT}.

In a contextual model, to use logical inferences one should be restricted at one particular context. If there are two propositions for which there does not exist any context containing both of them, then these two propositions cannot be compared or appear in the same logical inference. So far, as it is typically done, we have treated the properties corresponding to propositions and the possible contexts, independently of the dynamics of the theory. This position has been challenged recently by Sorkin \cite{Sorkin 2010} where he views logic as a dynamical entity. While it was not stated, a dynamic picture of logic was already adopted in consistent histories where the allowed (classical) questions, depend on the initial conditions and dynamics of the system. Before proceeding further, we should make some remarks about probabilities and their interpretation.

Interpreting the concept of probabilities has been an issue of debate since the early discussions of the founders of probability theory. It is outside the scope of this paper to delve in depth in this question. Here we will take the view that probabilities in general are epistemic in so far they concern with the knowledge (that an agent has) about the (ontic states of the) system. However, probability one or probability zero statements have a very different nature. As Einstein, Podolsky and Rosen stated in \cite{EPR} ``\emph{If, without in any way disturbing a system, we can predict with certainty (i.e., with probability equal to unity) the value of a physical quantity, then there exists an element of physical reality corresponding to this physical quantity}''. In other words, those type of propositions have ontic nature since they represent ``an element of physical reality'' of the system.

We now return to the logical structure and inferences one can make. The most important law of inference in order to reason deductively about a (possibly quantum) system is the Modus Ponens. It states that if $A\rightarrow B$ and $A$ is true, then $B$ is also true. Strictly speaking this and any other inference law should be applied to propositions that exist within a single (specified) context. However, if the statements in question are ontic in nature, as is the case with probability one propositions, then the truth values and the inferences derived by them should not depend on the context. This is reflected by the fact that in all CHVT the following three properties hold for the truth values that probability one and zero statements attain. Let a proposition $S$ be such that its probability is unity $p(S)=1$ which also implies that $p(\neg S)=0$\footnote{Here we identify logical propositions with a collection of subsets of a sample space (Kolmogorov). In this notation set inclusion $\subseteq$ means that there is a logical implication.}.

\begin{enumerate}

\item All contexts that include the proposition $S$, give truth value $v(S)=\textrm{True}$.

\item All contexts that include the negation of $S$, i.e. $\neg S$, give the truth value $v(\neg S)=\textrm{False}$.

\item If $S\subseteq S'$ then all contexts that include $S'$, give truth value $v(S')=\textrm{True}$.

\end{enumerate}
The importance of these conditions are in order to be able to use logical inferences. The first two imply that one can use ``Modus Tollens'' which is essential for proofs by contradiction, while the first and third imply that one can use ``Modus Ponens'' which is essential for deductive proofs. Those three properties, if we accept that probability one and zero statements are ontic, appear necessary to maintain a (contextual) realist view. The truth values of certain properties may vary depending on the context but this should not include properties that occur with certainty and are therefore elements of the reality. To sum up: (i) Physical predictions are separated to ``ontological predictions'' (probability one and zero) and ``epistemic predictions''. (ii) Statements and inferences involving ontological predictions should be context-independent, as is the case in standard quantum theory. As we will see in consistent histories if one does not supplement the formalism with some extra criteria, the third property is not satisfied.

Before proceeding further, we should give some definitions here. Assume we have propositions with corresponding projection operators $P$ and $Q$, corresponding to projections at one moment of time\footnote{A more general view of what a proposition is for a histories formulation will be given later}. If $[P,Q]\neq0$ they are called \emph{complementary}. If they are orthogonal and add to the identity, i.e. $P=1-Q$, and $PQ=QP=0$ they are called \emph{contradictory}. Finally, two propositions are called \emph{contrary} if they are orthogonal and not contradictory so that $P<1-Q$ and $PQ=QP=0$. A contrary inference is defined to be, when two contrary propositions are both implied with probability one. This is not possible in classical logic. In a contextual logic it is possible in principle to have contrary inferences, provided there does not exist any context containing both propositions. However, in standard quantum theory (and quantum logic), contrary inferences cannot exist because any two orthogonal propositions can be simultaneously tested by a measurement of a single observable (i.e. there is always a context that contains two orthogonal propositions).

\section{Consistent Histories Formalism}\label{Section consistent histories}

Consistent histories (or else decoherent histories) is an approach introduced by Griffiths, Omnes and Gell-Mann and Hartle (\cite{CH1,CH2,CH3,CH4}). The central aim of consistent histories is to assign probabilities to (coarse-grained) histories of a closed quantum system. Here we will give a brief review where the reader is  referred to the original papers for further details.

In histories formulations the central mathematical structure of interest, is the histories space $\Omega$, the space of all finest grained descriptions\footnote{Note, that following Hartle and Sorkin, we are adopting the path integral view of histories that takes the stance that there exist a unique (preferred) fine grained description, i.e. paths in the generalised configuration space. Other points of view, such as Isham's, are compatible with the one we take, at least in most ordinary cases.}. It is the set of all possible histories, and each element of it $h_i\in\Omega$ corresponds to a full description of the system, specifying every detail and property. For example, a fine grained
history gives the exact position of the system along with the specification of any internal degree of freedom, for every moment of time.

All physical questions correspond to subsets $A$ of the history space $\Omega$, where $A$ consists of a collection of fine grained histories. For example the question was the particle at position $x$ at time $t$ is the subset of $\Omega$ consisting of all trajectories that at time $t$ where at position $x$. We will refer to those questions as coarse grained histories. When we use the term ``history'' without further specification we refer to either a coarse grained or a fine grained history.

To each fine grained history, using the Feynman path integral we can assign a complex number (amplitude). This amplitude, depends on the initial state and on the dynamics of the system encoded in the action $S$:

\beq\label{Feynman} \a(h_i)=\exp i S(h_i)\eeq
Using this amplitude one can recover the transition amplitudes from $(x_1,t_1)$ to $(x_2,t_2)$ by summing through all the paths $\mathcal P$ obeying the initial and final condition:

\beq \a(x_1,t_1;x_2,t_2)=\int_\mathcal P \exp(iS[x(t)])\mathcal D x(t)\eeq
The mod square of this amplitude is the transition probability. To measure the interference between coarse grained histories we define the decoherence functional using the Feynman amplitudes Eq. (\ref{Feynman}):

\bea\label{path df} D(A,B)&=&\int_{A} \exp{(-iS[x(t)])}\mathcal{D}x[(t)]\int_{B} \exp{(+iS[y(t)])}\mathcal{D}y[(t)]\times\nonumber\\& &\delta(x(t_f)-y(t_f))\rho(x(t_0),y(t_0))\eea
Where $A$ and $B$ are any coarse grained histories, $t_f$ is the final while $t_0$ the initial moment of time considered and $\rho$ the initial state\footnote{The initial state can be viewed as the initial condition, which in principle is the initial state of the universe. The decoherence functional and the quantum measure that we will later define, depend crucially on this initial state. The decoherence functional encodes both the dynamics and the initial condition of the system}. The decoherence functional obeys the following conditions:

\begin{enumerate}
\item Hermiticity: $D(A,B)=D^*(B,A)$
\item Bi-linearity: $D(A\sqcup B, C)=D(A,C)+D(B,C)$, where $A\sqcup B$ is the disjoint union of $A$ and $B$.
\item Strong positivity\footnote{In the original references, a weaker condition was given, namely that $D(A,A)\geq 0$ which is called (weak) positivity, while the strong positivity condition was first defined in \cite{Sorkin QRW}. See Appendix \ref{Appendix strong positivity} for details.}: the matrix $D(A_i,A_j)$ is positive for any collection $\{A_1,A_2,\cdots,A_n\}$ of subsets of $\Omega$.
\item Normalisation\footnote{The sum in the expression should be replaced with an integral if we consider continuous histories.}: $\sum_{i,j}D(A_i,A_j)=1$
\end{enumerate}

The decoherence functional can also be defined using time ordered strings of projection operators. In particular

\beq\label{operator df} D(A,B)=\Tr(C^\dagger_A\rho C_B),\eeq
where $C_A$ and $C_B$ are the class operators, which are strings of time ordered projection operators corresponding to the (coarse grained) histories $A$ and $B$ respectively that we will specify below and is defined in the following way.

\beq C_A= P_{A_n}U(t_n-t_{n-1})\cdots U(t_3-t_2)P_{A_2}U(t_2-t_1)P_{A_1}\eeq
Here $U(t)$ is the unitary evolution operator that relates to the Hamiltonian via $U(t)=\exp (-iHt)$, and $P_{A_i}$ is the projection operator projecting at the subspace in which history A lay at time $t_i$. The history $A$ is the subset of $\Omega$ that contains all the histories that the system lies in the subspace that $P_{A_1}$ projects to, at time $t_1$ \emph{and} in the subspace that $P_{A_2}$ projects to, at time $t_2$, etc. The projection operators used for the definition of the class operator, in general, can project at any subspace of the Hilbert space (e.g. for point particle, the projections could project to either position or momentum space). Note that the expression for the class operator, is precisely the one used in ordinary quantum mechanics to obtain the amplitude, if some external observer carried out those measurements at the given times. By the linearity property of the decoherence functional Eq. (\ref{operator df}) can be extended to subsets of $\Omega$ that are not just strings of projection operators and are called inhomogeneous histories. We are not going into deeper discussion of the differences of those two definitions of the decoherence functional and their interpretational consequences. A final observation to make, is that from Eq. (\ref{operator df}) one can see that histories differing at the final moment of time automatically decohere because of the cyclic property of the trace. This also clarifies the source of the final time delta function in Eq. (\ref{path df}).

From the positivity condition, we can see that the diagonal elements of the decoherence functional are non-negative. Those terms are also referred to as quantum measure (\cite{quantum measure 1,quantum measure 2}) and are labelled as $\mu(A):=D(A,A)$.
The quantum measure cannot be interpreted as probability because due to interference, the additivity condition of probabilities fails:

\beq \mu(A\sqcup B)\neq\mu(A)+\mu(B)\eeq
The question of assigning probabilities to histories can now be rephrased in the following way. When is it possible to assign the quantum measure of a coarse grained history $A$ as the probability of this history $A$ actually occurring. The general aim of the approach is to be able to reason about a closed system with no reference to observer or an a-priori distinction of microscopic and macroscopic degrees of freedom, or a distinction between quantum and classical systems.

In order to make the quantum measure into a proper classical measure, one needs to restrict attention to some particular collection of subsets of $\Omega$ (coarse grained histories) rather than the full collection of all possible subsets of $\Omega$. The failure to satisfy the additivity condition can be traced at the off-diagonal terms of the decoherence functional as one can see from the very definition\footnote{Strictly speaking, from the real part of the off-diagonal terms.}. Let us take a partition of $\Omega$ which is defined to be a collection of coarse grained histories $\mathcal P_1=\{A^1_1,A^1_2,\cdots,A^1_n\}$ where $A^1_i\cap A^1_j=\emptyset$ and $\cup_i A^1_i=\Omega$. We call each coarse grained history $A^j_i$ of a partition $P_j$ as a cell of this partition. The superscript labels the partition considered, while the subscript labels the different cells of one partition. If for any pair of cells of one partition $A^1_i, A^1_j$ it holds that

\beq\label{decoherence condition} D(A^1_i,A^1_j)=0\textrm{ if }i\neq j\eeq
then the partition is called a \emph{Consistent Set} (CS). For this partition and any further coarse-graining, the standard rules of probability theory hold. The quantum measure, when restricted to those questions, becomes a classical measure. One would be tempted to assign these probabilities to the coarse-grained histories of the partition.

However, one can consider other partitions, say $\mathcal P_2=\{A^2_1,A^2_2,\cdots,A^2_n\}$. It is possible that this partition also forms a consistent set obeying Eq. (\ref{decoherence condition}). More importantly there does not exist, one finest-grained consistent set, that all other consistent sets arise as further coarse-grainings. The existence of multiple CS give rise to the following two problems.

The first problem is which CS to choose in order to make predictions. The probability rules apply to each CS separately so it is not possible to compare (coarse grained) histories that belong to different CSs. Without any further specification, one is forced in a contextual view of consistent histories, where each CS represents a different context. It is important to note here again, that since the set of possible CS (i.e. of contexts) depends on the initial conditions and dynamics of the system, the structure of the logic is also dynamical. The logical structure of such contextual view has been analysed using Topos theory by Isham \cite{isham-topos-histories} and subsequently by Flori \cite{Flori}. However, to make sense of such contextual-logic we need a suitable interpretation. It needs to be compatible with the physical requirements that a logical structure representing a physical system should obey. For example, such a physical requirement could be the compatibility of different contexts when making statements or inferences about ontic properties of the system. The viability of an interpretation depends exactly on which are these physical requirements and here we do not claim to answer this question in general. We concentrate on the physical requirement of compatibility of different contexts.
This leads to the second problem. Are the \emph{different} CS compatible with each other in that the physical (ontological) predictions that they make do not lead to contradiction.

Here we should stress that the quantum measure of a coarse grained history is independent of which CS it belongs and therefore any probabilistic predictions made using consistent histories do not depend on context/CS. The predictions of standard quantum mechanics are recovered. The problems arise when we consider the inferences that can be made that involve histories belonging to different CSs. In other words one is not allowed to make inferences for propositions belonging to different CSs even if those concern ``ontological predictions''. This is a form of contextuality stronger than the one in standard quantum theory, where inferences involving probability one predictions are independent of context. For example in standard quantum theory, if for a given state $\rho$ a proposition $P_1$ has probability unity ($\Tr(P_1\rho)=1$), then any coarser proposition $P_2$, such that $P_2 P_1=P_1$ occurs with certainty ($\Tr(P_2\rho)=1$) independent of the context (other projections that $P_2$ is measured along).

Because of the above issues, it has been argued that a complete interpretation of consistent histories should be accompanied with some principle that selects a, possibly unique, consistent set. In particular, the principle should at the very least address the second problem. In other words, if it does not provide a unique preferred CS, it should at least restrict the possible CSs in such a way that the remaining CSs are compatible with each other. This would (at least) allow for a realist albeit contextual view. The issue of a set selection criterion has been discussed already in the literature and one can identify two types of criteria. First are the criteria that are based on some physical principle (e.g. thermodynamic considerations) or induced by the environment \cite{Anastopoulos 1998,GeHa 2007,RZZ 2013}. Second are criteria that are based on the internal consistency of the logical structure (with respect to the physical predictions) (e.g. \cite{Kent 2000}). The criterion we will introduce in this paper is of the second type. In the remaining of the paper we will restrict attention to finite dimensional histories space.

\section{Contrary Inferences in Consistent Histories and Zero Covers}\label{Section contrary inferences}

In consistent histories, it is possible to have two contrary propositions $P,Q$ for a given moment of time $t$ and two consistent sets $C_1$ and $C_2$ such that

\EQ{C_1&=&\{h_P=P \textrm{ at time }t,h_{\neg P}=\neg P\textrm{ at time }t\}\textrm{ where }\mu(h_P)=1,\mu(h_{\neg P})=0\nonumber\\
C_2&=&\{h_Q=Q\textrm{ at time }t,h_{\neg Q}=\neg Q\textrm{ at time}t\}\textrm{ where }\mu(h_Q)=1,\mu(h_{\neg Q})=0}
In other words, in consistent histories there exist contrary inferences. Note however, that there does not exist any CS (context) that includes both $h_P$ and $h_Q$. We will give an example of contrary inference in consistent histories later.

A very important observation is that $h_{\neg P}\cup h_{\neg Q}=\Omega$, i.e. that the union of the negations of histories corresponding to contrary propositions, is the full histories space $\Omega$, i.e. these two histories generate an (overlapping) cover of $\Omega$. Moreover, each of these two histories have quantum measure zero. This is a very specific example of what is called a \emph{zero cover}. It is defined to be a collection of histories that  their union is $\Omega$ and each of these histories of the collection has quantum measure zero. For classical probability measures, no zero covers can possibly exist. However, for a quantum measure it is possible. Most of the conceptual problems in interpreting histories formulations arise due to this property. For example, the Kochen Specker theorem constitutes  an example of a zero cover \cite{DoGh 2008,SuWa 2010}.

Here we should stress the importance of quantum measure zero histories for consistent histories. Let a history $Z$ have quantum measure zero $\mu(Z)=0$. It is a property of the decoherence functional and thus of the quantum measure, that the negation of a quantum measure zero history $\mu(Z)=0$ has quantum measure one $\mu(\neg Z)=1$ and moreover the partition of $\Omega$ consisting of these two histories $\{Z,\neg Z\}$ forms a CS. In other words, \emph{all} quantum measure zero sets are contained in at least one CS. It is interesting to note however, that this is \emph{not} the case for quantum measure one histories, since there exist many such histories that do not belong to any CS \footnote{An example of the latter can be found in \cite{Wallden 2008} where the quantum measure of a particle never leaving a region of the configuration space is unity because of the quantum Zeno effect, without belonging to any consistent set. Note however, that if a quantum measure one history \emph{does} belong to a consistent set, then its negation has quantum measure zero.}.

\begin{lemma} Every zero cover that consists of two (coarse-grained) histories leads to a contrary inference.
\end{lemma}

\begin{proof} Let us assume that there is a two histories zero cover $\{Z_1,Z_2\}$, i.e. $Z_1\cup Z_2=\Omega$ and $\mu(Z_1)=\mu(Z_2)=0$. It follows that we have the following two CSs:

\EQ{C_3&=&\{Z_1,\bar Z_1:=\Omega\setminus Z_1\}\textrm{ where }\mu(Z_1)=0,\mu(\bar Z_1)=1\nonumber\\ C_4&=&\{Z_2,\bar Z_2:=\Omega\setminus Z_2\}\textrm{ where }\mu(Z_2)=0,\mu(\bar Z_2)=1}
We note that $\bar Z_1\cap\bar Z_2=\emptyset$ and also $\bar Z_1\subseteq Z_2$ and $\bar Z_2\subseteq Z_1$, i.e. $\bar Z_1$ and $\bar Z_2$ are contrary propositions. In the consistent set $C_3$ the history $\bar Z_1$ occurs with probability one, while in  consistent set $C_4$ the history $\bar Z_2$ occurs with probability one. We therefore see that every two histories zero cover leads to a contrary inference.
\end{proof}

\section{Examples of Zero Covers} \label{Section examples}

We will give two examples of zero covers. The first example is a three-slit experiment while the second is the ``two-site hopper''. For each example we will see the zero covers and what CSs exist. Only the first example contains a contrary inference (a two histories zero cover) but we give both for exploring the connection with coevents as it will become apparent later.

\subsection{Three slit}

Consider a set up, where we send a photon through a plate with three slits $A,B$ and $C$, and we consider (by post-selecting) a point on the final screen $P$, where the amplitude to go through slit $A$ or $C$ to $P$ is $+1$, while the amplitude to reach this point through slit $B$ is $-1$. We denote the ``fine grained'' history crossing slit $A$ and ending at $P$ as $h_{AP}$, etc. In other words we have $\alpha(h_{AP})=\alpha(h_{CP})=1$ and $\alpha(h_{BP})=-1$. We can see that $\mu(\{h_{AP}\})=\mu(\{h_{BP}\})=\mu(\{h_{CP}\})=1$. Moreover, we have the following coarse grained histories that have zero quantum measure $\mu(\{h_{AP},h_{BP}\})=\mu(\{h_{BP},h_{CP}\})=0$ and generate a zero cover

\EQ{\mathbf{Z}=\{ \{h_{AP},h_{BP}\},\{h_{BP},h_{CP}\}\}.}

For these amplitudes, one can see that there are three different consistent sets.

\begin{enumerate}

\item The set $\co_1=\{\{h_{AP},h_{BP}\},\{h_{CP}\}\}$, where $\mu(\{h_{CP}\})=1$ and $\mu(\{h_{AP},h_{BP}\})=0$.

\item The set $\co_2=\{\{h_{AP}\},\{h_{BP},h_{CP}\}\}$, where $\mu(\{h_{AP}\})=1$ and $\mu(\{h_{BP},h_{CP}\})=0$.

\item The set $\co_3=\{\{h_{AP},h_{BP},h_{CP}\}\}$, where  $\mu(\{h_{AP},h_{BP},h_{CP}\})=1$.

\end{enumerate}
It follows that $\co_1$ and $\co_2$ contain contrary propositions with probability one. In $\co_1$, the history $\{h_{CP}\}$ has probability one, while in $\co_2$ the history $\{h_{AP}\}$ which is contrary to $\{h_{CP}\}$, also has probability one. Intuitively, the one CS gives the particle crossing through slit $C$ with probability one, while the second CS gives the particle crossing slit $A$ with probability one. Moreover in the latter case the set $\{h_{BP},h_{CP}\}$ which incudes $h_{CP}$ has probability zero.

\subsection{Two-site hopper}

Here we consider a system that is known as the two-site hopper \cite{Sorkin 2013}. We have a single system, jumping between two sites $0,1$. The evolution of the system for a single time step is given by

\EQ{U=\frac1{\sqrt2}\left( \begin{array}{cc}
1 & i\\
i & 1\end{array} \right).}
In other words, the amplitude for the system to remain at the same site is $1/\sqrt2$ while to jump to the other site $i/\sqrt2$. Here we will consider the case of three moments of time, i.e. the state starts from site $0$, then we evolve it three times and therefore has $2^3=8$ fine grained histories. This example can be viewed as a discrete time example. However, note that it is equivalent with considering a qubit that one ``observes'' at three moments of time at a fixed basis, in which case one can find a Hamiltonian $H$ and time interval $t$ such that $U=e^{-iHt}$. We will use the following labeling of the fine grained histories (outcomes from right to left):

\begin{eqnarray}
\a(h_1=000)=\frac1{2\sqrt2}&,& \a(h_2=001)=-\frac1{2\sqrt2}\nonumber\\
\a(h_3=010)=-\frac1{2\sqrt2} &,& \a(h_4=011)=-\frac1{2\sqrt2}\nonumber\\
\a(h_5=100)=\frac1{2\sqrt2}i &,&\a(h_6=101)=-\frac1{2\sqrt2}i\nonumber\\
\a(h_7=110)=\frac1{2\sqrt2}i &,&\a(h_8=111)=\frac1{2\sqrt2}i
\end{eqnarray}
Note, that to compute the quantum measure one needs to take into account the $\delta$-function that exists in the decoherence functional at the final time. We can easily see that the following collection of coarse grained histories covers the full history space $\Omega$ and each of them has quantum measure zero

\EQ{\mathbf{Z}=\{\{h_1,h_2\},\{h_1,h_3\},\{h_1,h_4\},\{h_5,h_6\},\{h_6,h_7\},\{h_6,h_8\}\}.}
There are more quantum measure zero histories, generated by disjoint unions of the above mentioned sets.

\bea & &\{h_1,h_2\},\{h_1,h_3\},\{h_1,h_4\},\{h_5,h_6\},\{h_6,h_7\},\{h_6,h_8\},\{h_1,h_2,h_5,h_6\},\nonumber\\
& &\{h_1,h_2,h_6,h_7\},\{h_1,h_2,h_6,h_8\},\{h_1,h_3,h_5,h_6\},\{h_1,h_3,h_6,h_7\},\nonumber\\& &\{h_1,h_3,h_6,h_8\},
\{h_1,h_4,h_5,h_6\},\{h_1,h_4,h_6,h_7\},\{h_1,h_4,h_6,h_8\}
\eea
None of them has more than four (fine grained) histories. As we have mentioned earlier, in order to have contrary inferences one needs a two (coarse grained) histories zero cover. It follows that in this example there are no contrary inferences, even though there is a zero cover.

We can explicitly find the possible consistent sets. In particular, each measure zero coarse grained history generates a two coarse grained histories CS. Furthermore there are many other CSs such as

\EQ{\co_1&=&\{\{h_1,h_2,h_3,h_4\},\{h_5,h_6,h_7,h_8\}\}\textrm{ where }\mu(\{h_1,h_2,h_3,h_4\})=\mu(\{h_5,h_6,h_7,h_8\})=1/2\nonumber\\
\co_2&=&\{\{h_1,h_2,h_7,h_8\},\{h_3,h_4,h_5,h_6\}\}\textrm{ where }\mu(\{h_1,h_2,h_7,h_8\})=\mu(\{h_3,h_4,h_5,h_6\})=1/2\nonumber\\
\co_3&=&\{\{h_1,h_2\},\{h_3,h_4\},\{h_5,h_6\},\{h_7,h_8\}\}\nonumber\\
& &\textrm{ where }\mu(\{h_1,h_2\})=\mu(\{h_5,h_6\})=0,\mu(\{h_3,h_4\})=\mu(\{h_7,h_8\})=1/2}
In total there are 43 different CSs (including CSs that arise as coarse-graining of other CSs).

\section{Selection criterion: Preclusive Consistent Sets}\label{Section PCS}

To make sure that contrary inferences do not occur we should forbid the situation where a coarse-grained history that belongs to an allowed CS and has probability one, is contained in a quantum measure zero history. The above condition appears natural since, as argued in section \ref{Section contextuality}, ``ontological predictions'' and their inferences, should be non-contextual. We strengthen this requirement further, by requesting that no coarse-grained history that belongs to an allowed CS and has non-zero probability, is contained in a quantum measure zero history.

Probabilistic statements in general are epistemic and correspond to the knowledge an agent has about the system. When a given question $Z$ has quantum measure zero, it implies that this particular question is precluded which is a statement about the physical reality of the system. This means that anything contained in $Z$, e.g. $A\subseteq Z$ also does not occur, because there is no ontic state compatible with $A$. It would seem unreasonable for any agent to assign non-zero probability to question $A$ when there is no ontic state that is compatible with $A$. We would therefore need to rule-out contexts (CSs) that permit this. We now define preclusive consistent sets:\\

\noindent\textbf{Preclusive Consistent Set} (PCS): We call \emph{preclusive consistent set}, a consistent set $\co=\{A_1,A_2,\cdots\}$ that for all $i$ such that $\mu(A_i)\neq0$, there does not exist any history $Z\subseteq\Omega$ with quantum measure zero  $\mu(Z)=0$ such that $A_i\subseteq Z$.\\

We are now in position to define the set selection criterion that is the main result of this paper. We can assign a probability to a (coarse-grained) history, if it belongs to PCS. This leads to a contextual view, where the possible contexts are the different PCSs. Moreover, this condition ``selects'' from all the CSs only those that are preclusive. One can easily confirm that among the PCSs there are no contrary inferences. This set selection criterion appears to be the weakest criterion that one can introduce in order to avoid contrary inferences. It does not lead to a unique preferred CS but rules out CSs that are not compatible with the ontological predictions of the theory.

An important point to make here is that we do \emph{not} need to know all CSs in order to confirm whether one particular CS is preclusive. It suffices to check that no coarse grained history including (as subset) histories of the CS we consider have quantum measure zero\footnote{In this paper we have considered finite cardinality histories spaces. The case for infinite histories space is considerably more complicated, both conceptually and technically, since one needs to restrict attention to ``measurable'' sets. It is worth to note however, that in examples of discrete spacetimes, such as (finite) causal sets, the histories space is indeed finite.}. This is still a very complicated task, but simpler from other criteria. We will now examine the criterion named ``Ordered Consistent Sets'' (OCS) that Kent proposed in \cite{Kent 2000} and return in section \ref{Section discussion} discussing the different criteria.


\subsection{Kent's Ordered Consistent Sets}

Kent in \cite{Kent 2000} suggested a different set selection criterion, namely the Order Consistent Set (OCS). We need to give some definitions. A consistent history $H_i$ is a (generally coarse grained) history that belongs to at least one CS, $\exists\quad\mathcal {C}|H_i\in\mathcal{C}$. Here is important to note that any further coarse graining of a CS, is also a CS, so to confirm that a given history $H_i$ is consistent, it suffices to check that it decoheres with its negation, i.e. that the set $\{H_i,\Omega\setminus H_i\}$ is a CS. The set of all consistent histories (CH) is denoted $\mathcal{CH}$.

In the convention that we have adopted in this paper, all coarse grained histories are subsets of the history space $\Omega$. To define the OCS, we first need to define two partial orders on the space of CH's. First is the ``size partial order'' where a consistent history $H_i\in\mathcal{CH}$ is smaller than $H_j\in\mathcal{CH}$ (denoted $H_i \prec H_j$) if and only if $H_i\subseteq H_j$. Second is the ``measure partial order'' where a consistent history $H_i$ is smaller than $H_j$ (denoted $H_i<H_j$) if and only if $\mu(H_i)\leq \mu(H_j)$. While both these partial orders could be defined on all histories, following \cite{Kent 2000} we define them on the set $\mathcal{CH}$ only.

A consistent history $O_i\in\mathcal{CH}$ is \emph{ordered consistent history} if for all $H_j\in\mathcal{CH}$ the two ordering agree, i.e.

\bea
(1)\quad H_j\succ O_i&\Leftrightarrow& \mu(H_j)\geq \mu(O_i)\quad\forall \quad H_j\in\mathcal{CH}\nonumber\\
(2) \quad H_j\prec O_i &\Leftrightarrow& \mu(H_j)\leq\mu(O_i)\quad\forall \quad H_j\in\mathcal{CH}
\eea
An Ordered Consistent Set (OCS) is a consistent set that all the (coarse-grained) histories in the set are ordered consistent. The motivation for this criterion is to choose as physical CSs those that have histories with quantum measure compatible with the quantum measure of all other CHs including CHs belonging to other CSs. This will then lead to inferences ``independent of context''. The reason to restrict attention to the quantum measure of CHs rather than the quantum measure of all histories is because it is only the quantum measure of CHs that can be interpreted as probability (always within some CS).

In order to be able to compare the OCSs with the PCSs, we need to make the following observation. It can be shown that provided the decoherence functional is strongly positive (as we have assumed it is), if a history $A$ has quantum measure zero $\mu(A)=0$ then its negation $\neg A=\Omega\setminus A$ has quantum measure unity $\mu(\neg A)=1$ (see for example \cite{SuWa 2010}) and therefore they form a CS $\mathcal{C}_A=\{A,\neg A\}$. Therefore all precluded histories are CH in the terminology of \cite{Kent 2000}. It follows that the set of all precluded histories $\mathcal{Z}=\{Z_i\subset\Omega|\mu(Z_i)=0\}$ is a proper subset of the set of all consistent histories $\mathcal{CH}$.

It is now easy to see that an OCS is always PCS. An ordered consistent history $H_i$ that has non-zero quantum measure $\mu(H_i)=p\neq 0$ is not contained in any consistent history with quantum measure less than $p$. Since all the precluded histories are consistent histories, it follows that $H_i$ is not subset of any precluded history. Therefore an OCS is a CS that all the non-zero measure histories are not subsets of precluded histories and therefore all OCSs are also PCSs.

It is worth pointing out that the converse is not in general true, since it is possible for a history $A$ in a PCS to be subset of a consistent history $B$ that $\mu(B)<\mu(A)$. To illustrate this see an explicit example at the Appendix \ref{example: OCS-PCS}

In section \ref{Section discussion} we will expand further on the differences in motivation that the two principles have and on comparing them further by stressing their advantages and disadvantages.

\section{Coevents and Preclusive Consistent Sets}\label{Section coevents PCS}

Having a complete and conceptually satisfactory interpretation for a histories version of quantum theory is important. The consistent histories provide a framework to obtain probabilities for histories of a closed quantum system. We have addressed so far the potential inconsistencies within this framework. What it is not typically done in the consistent histories, is a discussion of what is the ontology behind the probabilistic statements. In other words while one can (attempt) to understand the probabilities given by the consistent histories as a state of knowledge of an agent,  we have not addressed the issue of what ontic states is the knowledge of this agent about.

The coevents formulation of quantum theory, introduced by Rafael Sorkin \cite{coevent1,coevent2} attempts to address the issue of which is the ontology of a histories formulation. While this formulation succeeds in circumventing the problems presented by the Kochen Specker theorem \cite{DoGh 2008,SuWa 2010} and manages to address the problem of ontology in a satisfactory manner, probabilistic predictions are difficult to recover. One is forced to go back at the founders of probability theory and use the (weak) Cournot principle in order to recover probabilistic predictions \cite{GhWa 2009}. An interesting possibility is that one may be able to use the machinery of consistent histories, which provides a full probabilistic calculus, for dealing with probabilistic predictions in the coevents formulation. This will also benefit the consistent histories approach, since it will obtain the underlying ontic states that the probabilities refer to.
Here we give some brief introduction to the relevant issues of this formulation, while further details can be found in this recent review \cite{coevents review}.

Ideally one would wish to have as possible ontic states of the theory the fine-grained histories. However, due to the existence of zero-covers every fine-grained history can be included  in a precluded (quantum measure zero) history. The alternative ontology suggested is that of a coevent. A coevent is a (coarse-grained) history $R$ that obeys two basic properties. (i) $R$ is preclusive, i.e. there does not exist $Z\in\Omega$ where $\mu(Z)=0$ and $R\subseteq Z$. (ii) It is minimal, i.e. given a coevent $R$ there does \emph{not} exist any $P\in\Omega$ where $P$ is preclusive and $P\subset R$. A coevent is the smallest possible preclusive history. The set of all coevents is the set of all preclusive, minimal histories and these are the possible ontic states of this formulation. It is worth noting, that the above definitions have an obvious classical limit. If instead of quantum measure we had a classical probability measure, then the above definitions give as ontic states fine grained histories which are the correct ontic states in classical physics. In other words the ontology of classical physics can be seen as the special case of this generalised ontology, and the difference arises only due to the different nature of the quantum measure.

Any given coevent $A$ gives in a truth value to all other coarse grained histories. In particular if the coevent $A$ is realised, history $B$ gets the value ``True'' if and only if $A\subseteq B$. In all other cases it gets truth value ``False''. We can easily see that contradictions exist if one considers a question $C$ such that both $C\cap A\neq\emptyset$ and $\neg C\cap A\neq\emptyset$. Here we should also stress, that given a quantum measure, one can find uniquely the set of all allowed coevents/ontic states, irrespective of whether they belong to a CS.

It would a failure of a theory if it allowed some ``agents'' to make a prediction that some property occurs with non-zero probability, but there does not exist a single ontic state that is compatible with this property. This was the case for standard consistent histories but not for PCSs. From the very definition of a coevent we can see the following. Any history $A$ that belongs to a PCS and has non-zero quantum measure is a preclusive history. It follows that either $A$ is a coevent, or $A$ contains a coevent. This means that any PCS is compatible with the coevents ontology. For any epistemic statement that assigns a non-zero probability to a history $A$, there is at least one ontic state (coevent $B$) that is compatible with this history, meaning that is fully contained within a single (coarse grained) history\footnote{This compatibility means that if the coevent $B$ is realised, then the statement ``history $A$ occurred'' gets truth value ``True'', while the negation $\neg A$ gets truth value ``False''.}.

\section{Examples revised}\label{Section examples 2}
We now reexamine the examples in view of the selection criterion brought forward and the potential coevents ontology.
\subsection{Three slit}
It easy to confirm that from the three CSs that existed, it is only the third CS that is PCS. In $\co_1$ it is the history $\{h_{CP}\}$ that has non-zero quantum measure and is subset of the quantum measure zero history $\{h_{BP},h_{CP}\}$. In $\co_2$ it is the history $\{h_{AP}\}$ that has non-zero quantum measure and is subset of the quantum measure zero history $\{h_{AP},h_{BP}\}$. This is a very special example, where the requirement of PCS leads to unique CS which in this case is the trivial CS, namely $\co_3$.

There is also a unique allowed coevent in this example (also a very specific feature of this example). This coevent is the history $\{h_{AP},h_{CP}\}$, since it is a preclusive history and is not contained at any other preclusive history. As one can see, it is fully contained at a single history of the only allowed CS  (and therefore leads to no contradiction).

\subsection{Two-site hopper}

As we have already mentioned, there are no contrary inferences in this example. Furthermore, one can show that \emph{all} CS of this example are PCS. To see this we should note that to rule out any CS it needs to contain a non-zero history that is subset of a quantum measure zero history. The largest (in terms of cardinality) quantum measure zero history of this example contains four fine-grained histories. One can see that the smallest non-zero history included at any CS contains four fine-grained histories and therefore cannot be subset of \emph{any} quantum measure zero set. We see that in this example not only we are not lead to a unique CS, but the proposed criterion does not rule out any CS.

This example, also allows for six possible coevents. These are

\beq
\{\{h_2,h_3\},\{h_2,h_4\},\{h_3,h_4\},\{h_5,h_7\},\{h_5,h_8\},\{h_7,h_8\}\}
\eeq
One can verify that in all CSs, histories that have non-zero quantum measure contain at least one of these six coevents.

\section{Discussion}\label{Section discussion}

We have identified as one of the major problems of the consistent histories approach the fact that the possible contexts (CSs) allowed by the formalism lead to contextual ontological inferences. An example of this problem is the existence of contrary inferences. As it has been argued earlier, when dealing with probability one and probability zero statements (ontological predictions) one should obtain non-contextual truth values and non-contextual inferences. We aimed at restricting the possible CSs, by introducing a set selection criterion. We requested that only \emph{preclusive consistent sets} should be considered. Within the PCSs no contrary inferences exist and moreover the ontological predictions and inferences become non-contextual.

Starting with different motivation, Adrian Kent \cite{Kent 2000} suggested a different criterion, namely that an allowed CS should consist of \emph{ordered} histories. The motivation was that the quantum measure of any history belonging to a physically allowed CS should be compatible with the quantum measure of all CHs. We on the other hand, in our motivation for the PCSs were concerned with the compatibility of logical inferences and in particular those that arise from ontological predictions. Both criteria resolve the issue of contrary inferences by making sure that CSs that contain histories that may lead to contrary inferences are not physically allowed CSs.

The OCSs was shown to be strictly more restrictive from PCSs. If one is looking for a criterion that limits the physically allowed CSs as much as possible, approaching the ``ideal'' case where a single finest grained CS exists, then Ordered Consistency is a better criterion than Preclusive Consistency. However, there are two major reasons why this may not be case. The first reason is that while we want a restrictive criterion, we do not want a too strong criterion. In other words we do not want the criterion to allow only for trivial CSs, or to exclude a CS that intuitively corresponds to a semiclassical situation that concerns macroscopic observables. It is therefore desirable to find the weakest criterion that rules out contrary inferences. This could be later complemented (if one wished) with some stricter criterion (possibly leading to a unique CS) provided that the latter comes with a proof that it allows for histories that include the semiclassical evolution of macroscopic objects.

The second and possibly more important advantage of PCSs has to do with the simpler definition it has.  In particular there is a conceptual and a technical advantage of having a simpler definition. In Preclusive Consistency there are two types of histories, those that belong to a PCSs and those that do not. Moreover to decide if one CS is PCS, one needs to check that each of the histories belonging to the CS are not subset of one of the Precluded Histories (i.e. checking that all the supersets are not preclusive). To compute if a history is preclusive, one needs only the quantum measure of this history. In Ordered Consistency there are three types of histories. There are the histories that do not belong to any CS, there are the histories that do not belong to any OCS but belong to a CS and are used to define the OCS  and finally there are the histories that belong to an OCS that are used for probabilistic predictions. The status of histories that are used for defining the OCS but are not themselves members of an OCS is strange. If one does not consider this CS as physical, why to give the histories belonging to a CS that is not physical a different status (and importance) than histories that do not belong to any CS at all? This also leads to a practical problem as well. In order to decide if a CS is OCS one needs to first compute all the CHs (a task clearly more difficult than computing the precluded histories) and check that the quantum measure of the histories in the CS in question are compatible with all CHs by checking all the supersets and all the subsets of each history. Therefore given a CS it is much more difficult to decide if it was an OCS than if it was a PCS\footnote{Note that already deciding if a CS is PCS is a very demanding task even for history space $\Omega$ of small cardinality.}.

Here we should step back, and by giving a precise picture of the difficulties that consistent histories face in order to form a fully satisfactory interpretation of closed system quantum theory, we can evaluate what has been achieved by Preclusive Consistency and what still remains unanswered.

The motivation to use a histories formulation comes from the desire to address multiple time propositions and also to use a formulation of quantum theory that space and time appear in more equal footing which would suit a quantum theory of gravity. However, due to problems as the one presented in \cite{Kent 1997}, consistent histories formulation seems to suffer from problems more severe than what standard QT faces. A first step to construct a fully satisfactory theory of closed quantum systems, is to lift those inconsistencies. A second step would be to provide the correct ontology and possibly remove further inconsistencies that exist in standard QT. Here we mainly aim to address the first step.

In particular in consistent histories there exist multiple CSs/contexts. This is not different (better or worse) from standard QT. For example the position and momentum measurements in single-time QT are different ``contexts'' that cannot be compared. The probabilities predicted by QT are independent of which context (basis) each proposition belongs and this is the same in consistent histories. The crucial point where consistent histories are doing worse than standard QT, is that in standard QT if a proposition $A$ belongs to more than one context and has probability one, then all other propositions that are coarse grainings of proposition $A$, also happen with probability one in any context that they belong. This is not the case in consistent histories, because of the existence of contrary inferences. By the principle of Preclusive Consistency one evades this problem.

A further aim would be to find a criterion that evades all possible incompatibilities, including those involving probabilities between zero and one. Moreover, ideally we would like to end up with a unique CS that also recovers classical evolution for macroscopic objects. However this would be part of the second step described earlier, which is resolving problems that exist in standard QT (such as the existence of incomparable contexts and the unambiguous emergence of classical physics).

Part of the second step is also to find the suitable ontic states. In this paper we mention the coevent formulation, which gives one candidate for ontic states. This was done mainly to illustrate that removing the strong inconsistencies such as the existence of contrary inferences allows for ontological descriptions that are compatible with PCSs.

Here we should stress that although the criterion of preclusivity may be very desirable for any candidate consistent set to satisfy, there is still a tremendous amount of remaining ambiguity in identifying a unique (or essentially unique) set.  Although they will not contain contrary inferences, the various PCSs will still be largely incomparable to each other, and most will have no interpretation in terms of well known macroscopic classical variables.

A final issue to analyse, is the connection of PCSs with the coevents formulation. We found in section \ref{Section coevents PCS} that PCSs are compatible with coevents. It appears that a combination of coevents with consistent histories could be positive for both formulations. On the one hand consistent histories would obtain a possible underlying ontology for which the formalism provides predictions. On the other hand the coevents formulation will obtain an alternative way dealing with probabilistic predictions and will be able to benefit from the developed formalism of the consistent histories. Furthermore one could attempt to use more properties of the coevents formulation to restrict the possible CSs further, always in a way compatible with the coevent ontology. In the coevent formulation, one can define a unique finest-grained classical domain, the ``principal classical partition'' (see appendix of \cite{GhWa 2009}). Requiring that allowed CSs should arise as a coarse graining of the principal classical partition would give another criterion but we leave the analysis of such a criterion for future work\footnote{We should note here, that similarly with the OCSs it would be much more difficult to decide if a given CS satisfies this new criterion.}.

\noindent \emph{Acknowledgments}: This work was supported by the UK Engineering and Physical Sciences Research Council (EPSRC) through grant EP/K022717/1. Partial support from COST Action MP1006 is also gratefully acknowledged. The author thanks the anonymous reviewers for their constructive comments.

\appendix

\section{Weak Vs Strong Positivity of Decoherence Functional}\label{Appendix strong positivity}

In section \ref{Section consistent histories} we gave the definition of the decoherence functional as either path integral or string of projection operators. Moreover we gave four conditions that the defined decoherence functional obeys. The third condition was the positivity condition and we noted that there exist two versions of this condition. The stronger one was used in this paper while the weaker one was given in the first papers \cite{CH1,CH2,CH3,CH4}. One can show that for decoherence functionals defined by equations (\ref{path df}) and (\ref{operator df}), i.e. for standard quantum mechanics, the strong positivity condition is indeed obeyed.

It is important to show that the strong positivity condition is satisfied by all decoherence functionals allowed, because this is crucial in determining the relation between PCSs and OCSs.

When the initial state is pure, the decoherence functional for two histories $\alpha$ and $\beta$ can be expressed as the inner product between the branch state vectors $\psi_\alpha$ and $\psi_\beta$.  But this means the decoherence functional over any set of histories $\{\alpha\}$ is a Gram matrix, which is automatically a positive matrix.  When the initial state is impure, it is just a mixture of pure state so that the decoherence functional is a mixture of (positive) Gram matrices, which must also be positive.

In consistent histories framework as considered by Omnes, Griffiths, and Gell-Mann and Hartle, branch state vectors are obtained by applying the class operators -- i.e. strings of time-evolved projectors -- to the initial state.  Importantly, this does \emph{not} rely on having a set $\Omega$ of fine-grained histories from which all other histories are a coarse-graining.

However one could, instead of starting with definitions of Eq. (\ref{path df}) and Eq. (\ref{operator df}), take a more radical view and use conditions (1-4) of section \ref{Section consistent histories} as defining the decoherence functional. This view can be adapted in order to \emph{generalise} quantum mechanics in such a way that is desirable for cases that time (and thus time ordered projections) may not be well defined, as for example in certain approaches to quantum gravity. Another reason that this type of generalisations are interesting is because they do not presuppose the Hilbert space structure.

In attempting to generalise quantum mechanics one is free to choose to use either the strong or the weak positivity condition. In \cite{DJS 2010} it was shown that if one adopts the strong positivity condition, then some Hilbert space structure can be recovered starting solely from conditions (1-4). Interestingly, in \cite{DHW 2014} it was shown that even if one is restricted to strongly positive decoherence functionals, there are some correlations allowed that are not predicted by standard quantum theory. Therefore, even the stronger positivity condition constitutes a generalisation of standard quantum mechanics. Moreover in \cite{DHW 2014} it was also pointed out that the weak positivity condition is not closed under composition. In other words it is possible to have two uncorrelated, non-interacting systems described by weak positive decoherence functionals and if one attempted to construct a joint decoherence functional for those systems, it no longer obeys the positivity condition. On the other hand, strongly positive decoherence functionals are closed under composition (two systems described by strongly positive decoherence functionals lead to a joint decoherence functional that is also strongly positive).

For the above reasons, if one was to generalise quantum mechanics using conditions (1-4) as starting point, then there are good physical reasons to believe that it is more appropriate to adopt the \emph{strong} positivity condition.

\section{Example of PCS that is not OCS}\label{example: OCS-PCS}

Here we will give an explicit example that a CS is Preclusive but not Ordered. Let us consider an example with three possible histories $\Omega=\{h_1,h_2,h_3\}$. Consider the following decoherence functional

\EQ{D=\left( \begin{array}{ccc}
1/3 & -7/24 & 7/24\\
-7/24 & 1/2 & -7/24\\
7/24 & -7/24 & 3/4\end{array} \right).}
One can easily check that it obeys the requirements of a strongly positive decoherence functional (is a positive matrix, symmetric, normalised to unity with all diagonal terms non negative). There are two CS, namely $C_1=\{\{h_1\},\{h_2,h_3\}\}$ which leads to the coarse grained decoherence functional

\EQ{D=\left( \begin{array}{cc}
1/3 & 0\\
0 & 2/3\end{array} \right).}
and $C_2=\{\{h_1,h_2\},\{h_3\}\}$ which leads to the coarse grained decoherence functional

\EQ{D=\left( \begin{array}{cc}
1/4 & 0\\
0 & 3/4\end{array} \right).}
Since there is no precluded history, both CSs are PCSs. However, $C_1$ is not an OCS since history $\{h_1\}$ has quantum measure $\mu(\{h_1\})=1/3$ which is greater than the quantum measure of $\{h_1,h_2\}$ which is $\mu(\{h_1,h_2\})=1/4$. Both $\{h_1\}$ and $\{h_1,h_2\}$ are consistent histories and (evidently) $\{h_1\}\subset\{h_1,h_2\}$. Similarly the consistent set $C_2$ is also not an OCS. This means that the only OCS is the trivial CS, while both $C_1$ and $C_2$ are PCSs.

 \end{document}